\documentclass[reqno]{amsart}

\usepackage[margin=0.7in]{geometry} 
\usepackage{amsmath,amsthm,amssymb,scrextend,bm,verbatim,fancyhdr, hyperref, enumitem, multirow, multicol, subfigure, xcolor, graphicx}
\usepackage[section]{placeins}

\newtheorem{thm}{Theorem}

\renewcommand{\qed}{\hfill$\blacksquare$\\}

\renewenvironment{proof}{\begin{addmargin}[1em]{0em}\begin{newproof} }{\end{newproof}\qed\end{addmargin}}
\numberwithin{equation}{section}

\newtheoremstyle{lemmadd}{\topsep}{\topsep}{\mdseries}{0pt}{\bfseries}{. }{ }{\thmname{#1}\thmnumber{ #2}\textnormal{\thmnote{ (#3)}}}
\theoremstyle{lemmadd}

\newtheoremstyle{propadd}{\topsep}{\topsep}{\mdseries}{0pt}{\bfseries}{. }{ }{\thmname{#1}\thmnumber{ #2}\textnormal{\thmnote{ (#3)}}}
\theoremstyle{propadd}

\newtheoremstyle{theoremdd}{\topsep}{\topsep}{\mdseries}{0pt}{\bfseries}{. }{ }{\thmname{#1}\thmnumber{ #2}\textnormal{\thmnote{ (#3)}}}
\theoremstyle{theoremdd}

\newtheoremstyle{definition}{\topsep}{\topsep}{\mdseries}{0pt}{\bfseries}{. }{ }{\thmname{#1}\thmnumber{ #2}\textnormal{\thmnote{ (#3)}}}
\theoremstyle{definition}

\newtheoremstyle{propositiondd}{\topsep}{\topsep}{\mdseries}{0pt}{\bfseries}{. }{ }{\thmname{#1}\thmnumber{ #2}\textnormal{\thmnote{ (#3)}}}
\theoremstyle{propositiondd}

\newtheoremstyle{corollarydd}{\topsep}{\topsep}{\mdseries}{0pt}{\bfseries}{. }{ }{\thmname{#1}\thmnumber{ #2}\textnormal{\thmnote{ (#3)}}}
\theoremstyle{corollarydd}

\newtheoremstyle{example}{\topsep}{\topsep}{\mdseries}{0pt}{\bfseries}{. }{ }{\thmname{#1}\thmnumber{ #2}\textnormal{\thmnote{ (#3)}}}
\theoremstyle{example}

\title{Option Pricing with Stochastic Volatility, Equity Premium, and Interest Rates}

\author{Nicole Hao$^1$, Echo Li$^2$, Diep Luong-Le$^3$}

\date{\today %
    }
    \thanks{1. Department
 of Mathematics,  Cornell University,
Ithaca, NY 14853, USA, (yh397@cornell.edu)}
    \thanks{2. Department
 of Mathematics,  Ohio State University,
Columbus, OH 43210, USA, (li.11789@buckeyemail.osu.edu)}
    \thanks{3. Department
 of Mathematics,  Lehigh University,
Bethlehem, PA 18015, USA, (dll224@lehigh.edu)}
 

\begin{document}
\maketitle

\begin{abstract}
\noindent 
This paper presents a new model for options pricing. The Black-Scholes-Merton (BSM) model plays an important role in financial options pricing. However, the BSM model assumes that the risk-free interest rate, volatility, and equity premium are constant, which is unrealistic in the real market. To address this, our paper considers the time-varying characteristics of those parameters. Our model integrates elements of the BSM model, the Heston (1993) model  for stochastic variance, the Vasicek model (1977) for stochastic interest rates, and the Campbell and Viceira model (1999, 2001) for stochastic equity premium. We derive a linear second-order parabolic PDE and extend our model to encompass fixed-strike Asian options, yielding a new PDE. In the absence of closed-form solutions for any options from our new model, we utilize finite difference methods to approximate prices for European call and up-and-out barrier options, and outline the numerical implementation for fixed-strike Asian call options.  
\end{abstract}

\section{Introduction} \label{intro}
\quad In this paper, we introduce a comprehensive model to price options expanding upon the Black-Scholes \cite{BS,M} model, by integrating the Heston model  \cite{heston_closed-form_1993} for a time-varying variance on the stock, the Vasicek \cite{Vasicek} model for a time-varying interest rate, and the Campbell-Viceira model \cite{CV1, CV2}  for a time-varying equity premium. 
In particular, we consider the following system of stochastic differential equations (SDEs)
\begin{align} \label{mainmodel}
    \begin{cases}
        dS(t)  = (\mu + X(t) +R(t))S(t)dt + \sqrt{\sigma_s(t)}S(t)dW_1(t)\\
        dX(t) = -\kappa_xX(t)dt + \sigma_x ( \rho_xdW_1(t)+\sqrt{1-\rho_x^2}  dW_2(t))\\
        d\sigma_s(t) = \kappa_s(\sigma - \sigma_s(t))dt + \eta \sqrt{\sigma_s(t)}(\rho_sdW_1(t)+\sqrt{1-\rho_s^2}dW_3(t))\\
        dR(t) = \kappa_r (r-R(t))dt + \sigma_r (\rho_r dW_1(t)+\sqrt{1-\rho_r^2}dW_4(t)).\\
    \end{cases}
\end{align}
\quad The random variables $S(t), X(t), \sigma_s(t),$ and $R(t)$ represents the underlying asset price, the deviation in the equity premium from its mean, the volatility, and the risk-free interest rate at time $t$. The parameters $\mu, \sigma_x, \sigma, \eta,$ and $\sigma_r$ are the long term average equity premium on the stock, the volatility of the equity premium, the long term average volatility of the stock, the volatility of the volatility, and the volatility of the interest rate. The $W_1, W_2, W_3,$ and $W_4$ are independent Brownian motions on a probability space $(\Omega, \mathcal F, \mathbb P)$ adapted to a filtration $\mathcal F_t$. The $\rho_x, \rho_s$, and $\rho_r$ are the correlation between the stock price and the change in equity premium, between the stock price and the volatility, and between the stock price and the interest rate; this formation assumes that these processes are correlated only through the stock price process itself, which may not be unreasonable since they are not directly observable through market data. The parameters $\kappa_x, \kappa_s$, and $\kappa_r$ correspond to the pressure for the equity premium,  volatility, and  interest rate to return to their long term average. 

We also extend our model from vanila options to price fixed-strike Asian call options. Since the payoff of Asian options depends on the running average of the underlying asset, we need to add a time-varying variable to the model in \ref{mainmodel}. We let $I$ denote the sum of underlying asset price over the time period from starting time, $T_0$,  to time $t$ 
\begin{equation} \label{Asianeq}
   I(t) = \int_{T_0}^t S(\tau) d\tau \quad \Longleftrightarrow \quad   dI(t) = S(t) dt.
\end{equation}
\quad To formulate the market model, we make foundational assumptions to ensure the feasibility.
\begin{itemize}
    \item The market is arbitrage-free. For two assets $P$ and $V$, if $P(T) = V(T)$, then $P(t) = V(t) \text{ } \forall \text{ } 0\leq t \leq T$.
    \item All processes $S$, $\sigma_s$, $X$, and $R$ are \textit{pricing processes} (see appendix for the definition of pricing processes). 
    \item There are no transaction costs when trading assets.
    \item The market is perfectly liquid. Traders are allowed to purchase or sell any amount of stock at any given time.
\end{itemize}

While there are studies extending Heston model, such as Grzelak and Oosterlee (2010) \cite{grzelak_cross-currency_2010} on the Heston model with stochastic interest rates, there is no literature integrating all of the factors we consider. The model \ref{mainmodel} extends the Heston stochastic volatility model (Heston, 1993) \cite{heston_closed-form_1993} by incorporating stochastic change in equity premium from Campbell and Viceira (1999, 2002) and interest rate from Vasicek (1977). 
This enhancement addresses the limitations associated with assuming constant values for these parameters. The change in equity premium, and the interest rate follow Ornstein–Uhlenbeck process since because we allow these processes to take negative values. We choose the Vasicek model over the Cox-Ingersoll-Ross model because we wish to allow for negative interest rates. Moreover, previous authors are interested in finding closed form solutions; we are interested in efficient numerical algorithms for estimating the value of the option by deriving a partial differential equation (PDE) the value satisfies, and estimating the solution after imposing boundary data. 

Our model assumes an incomplete market, since we assume only the stock and the risk-free asset are tradeable. This issue poses challenges when deriving a PDE and formulating an initial value boundary problem. We resolve the issue by treating all pricing processes, including $\sigma_s$, $X$, and $R$, as tradeable assets, effectively completing the market.

To showcase the practical utility of our formulated model, we estimate solutions to our PDEs by imposing terminal and boundary conditions for both European call and knock-out barrier options. We implement three numerical methods: forward Euler, backward Euler, and the Crank-Nicolson schemes. Subsequently, we compare numerical results yielded by these schemes and provide evidence of their convergence. 
The paper's structure is outlined as follows: In Section \ref{result}, we introduce and provide the derivation of new PDEs for our model.
Section \ref{numerical} presents boundary conditions for European call option, up-and-out barrier option, and Asian fixed-strike option along with numerical estimates for European call and up-and-out barrier options. In Section 4, we summarize key findings and their implications.

\section{Main results} \label{result} 

In this section we derive new PDEs for option pricing using two approaches:  replicating portfolio theory and  risk-neutral pricing. 

Let $V$ denote the price of an option on a stock, $S$, modeled by the system  \ref{mainmodel}. 
We denote $V_s$, $ V_{\sigma_s}$, and $V_x$ the first-order derivative of $V$ with respect to the underlying asset price $S$, variance of the underlying asset $\sigma_s$, and the change in equity premium $X$. Let $V_{S \sigma_s}, V_{SR}, V_{R\sigma_s}, V_{X\sigma_s}, V_{XS},$ and $V_{XR}$ be the mixed second-order partial derivatives. 
Our first main result concerning the price of the options as summarized in the following theorem.
\begin{thm} \label{ourPDE}
The price of a European style derivative on a stock price process determined by system \ref{mainmodel} must be a solution to the following PDE  
\begin{align} \begin{split}  \label{MainPDE}
    \quad \ V_t &= R(V-SV_s-XV_x-\sigma_sV_{\sigma_s}-RV_r)
    - \frac{1}{2}\sigma_sS^2 V_{ss} - \frac{1}{2}\sigma_x^2 V_{xx} - \frac{1}{2}\eta^2\sigma_s V_{\sigma_s\sigma_s}- \frac{1}{2}\sigma_r^2V_{rr}\\
    &\quad - \rho_x\sigma_x\sqrt{\sigma_s}SV_{sx} - \rho_s\eta\sigma_sSV_{s\sigma_s} - \rho_r\sigma_r\sqrt{\sigma_s}SV_{sr} \\&\quad - \rho_x\rho_s\sigma_x\eta\sqrt{\sigma_s}V_{x\sigma_s} - \rho_x\rho_r\sigma_x\sigma_rV_{xr} - \rho_s\rho_r\eta\sigma_r\sqrt{\sigma_s}V_{\sigma_sr}.
\end{split}
\end{align}
\end{thm}
\begin{proof}
We use two approaches to derive the PDE: replicating portfolio and change of measure. For the first approach, we will treat our market as a complete market. In other words, all assets, including stock, change in equity premium, volatility, and interest rate, are tradeable.

    \noindent {\textbf{\\Replicating portfolio approach:}}
    We consider a mathematical economy consisting of a stock and a bond. Investors may construct a portfolio, which consists of investments in these two vehicles. We assume the bond pays the risk-free interest rate, $r$, so that any investment, $B_0$, in the bond grows according to \[B(t)=B_0e^{rt} \iff dB = rBdt, \ B(0) = B_0. \]
    
    We assume the assets follow the model \ref{mainmodel}. A portfolio denote by $P$ begins with $P(0)$ dollars at time $t=0$. The agent may form a portfolio of consisting of a bond, $\Delta_s$ shares of stock at time $t$ for a cost of $\Delta_sS$ dollars, $\Delta_{\sigma_s}$ shares of variance at time $t$ for a cost of $\Delta_{\sigma_s}\sigma_s$ dollars, $\Delta_x$ shares of equity premium at time $t$ for a cost of $\Delta_xX$, and $\Delta_r$ shares of interest rate at time $t$ for a cost of $\Delta_rR$. $\Delta_s, \Delta_{\sigma_s}, \Delta_x,$ and $\Delta_r$ may be any adapted stochastic process. 
    
    The remainder of the money in the portfolio will be invested in the bond at the risk-free rate. Thus, the change in the value of the portfolio is
\begin{align} \label{dP}
    dP&=R(P-\Delta_sS-\Delta_xX-\Delta_{\sigma_s}\sigma_s-\Delta_rR)dt+\Delta_sdS + \Delta_{\sigma_s}d\sigma_s + \Delta_xdX + \Delta_rdR
\end{align}

Next, we introduce a derivative security, its value at time $t$ will be denoted $V(t).$ We will assume that the contract of the derivative specifies that it can be exercised at time $T>0,$ and that the value of the derivative at time $T$ depends only upon $S(T),$ $V(T)=f(S(T)).$ Our goal is to find the value, or the price of $V(t)$ for $T>t\geq 0$. 

Now, using the no arbitrage principle, we will construct a replicating portfolio. In particular, we will find a portfolio which satisfies $P(T) = V (T)$, and therefore, the amount in the portfolio at any earlier time, $P(t)$, must be the value of the derivative, $V (t)$. That $P(t) = V(t)$ for all $0\leq t\leq T$ is equivalent to $dP = dV$ and $P(T) = V(T)$. 
Using Ito's lemma, we have
\begin{align}\begin{split} \label{dV}
    dV(t, S, X, \sigma_s, R) &= V_tdt + V_sdS+ V_xdX + V_{\sigma_s}d\sigma_s + V_rdR \\
    &\quad + \frac{1}{2}V_{ss}dSdS+\frac{1}{2}V_{xx}dXdX + \frac{1}{2}V_{\sigma_s\sigma_s}d\sigma_sd\sigma_s + \frac{1}{2}V_{rr}dRdR\\
    &\quad + V_{sx}dSdX + V_{s\sigma_s}dSd\sigma_s + V_{sr}dSdR + V_{x\sigma_s}dXd\sigma_s + V_{xr}dXdR + V_{\sigma_sr}d\sigma_sdR.
\end{split}
\end{align}

Notice that only $dX, d\sigma_s,$ and $dR$ terms in both $dP$ and $dV$ have $dW_2, dW_3$, and $dW_4$ respectively; and $dW_1$ term appears in $dS, dX, d\sigma_s,$ and $dR$. Therefore, setting $dP = dV$, we find that from $dW_1, dW_2, dW_3,$ and $dW_4$ that
\begin{align*}
        \Delta_s + \Delta_{\sigma_s}+ \Delta_x + \Delta_r &= V_s + V_x + V_{\sigma_s} + V_r &&\quad \text{ (from } dW_1  \text{ term), }\\
        \Delta_x &= V_x &&\quad \text{ (from } dW_2  \text{ term), }\\
        \Delta_{\sigma_s} &= V_{\sigma_s} &&\quad \text{ (from } dW_3  \text{ term), }\\
        \Delta_r &= V_r &&\quad \text{ (from } dW_4  \text{ term). }
\end{align*}

Thus, we also get $ \Delta_s = V_s$. Plugging $ \Delta_s = V_s$, $ \Delta_x = V_x$, $ \Delta_{\sigma_s} = V_{\sigma_s}$, and $ \Delta_r = V_r$ into equation \ref{dP}, we get
\begin{align} \label{dP2} 
    dP&=R(P-V_sS-V_xX-V_{\sigma_s}\sigma_s-V_rR)dt+V_sdS + V_{\sigma_s}d\sigma_s + V_xdX + V_rdR .
\end{align}

Equating $dP=dV$ from \ref{dV} and and \ref{dP2} and plugging $P = V$, we get
\begin{align*}
    & \ R(V-V_sS-V_xX-V_{\sigma_s}\sigma_s-V_rR)dt+V_sdS + V_{\sigma_s}d\sigma_s + V_xdX + V_rdR \\
    &= V_tdt + V_sdS+ V_xdX + V_{\sigma_s}\sigma_s + V_rR + \frac{1}{2}V_{ss}dSdS+\frac{1}{2}V_{xx}dXdX + \frac{1}{2}V_{\sigma_s\sigma_s}d\sigma_sd\sigma_s + \frac{1}{2}V_{rr}dRdR\\
    &\quad + V_{sx}dSdX + V_{s\sigma_s}dSd\sigma_s + V_{sr}dSdR + V_{x\sigma_s}dXd\sigma_s + V_{xr}dXdR + V_{\sigma_sr}d\sigma_sdR,
\end{align*}
which gives 
\begin{align*} \begin{split} 
      \ R(V-V_sS-V_xX- & V_{\sigma_s}\sigma_s-V_rR)dt 
    = V_tdt + \frac{1}{2}V_{ss}dSdS+\frac{1}{2}V_{xx}dXdX + \frac{1}{2}V_{\sigma_s\sigma_s}d\sigma_sd\sigma_s + \frac{1}{2}V_{rr}dRdR\\
    & \quad + V_{sx}dSdX + V_{s\sigma_s}dSd\sigma_s + V_{sr}dSdR + V_{x\sigma_s}dXd\sigma_s + V_{xr}dXdR + V_{\sigma_sr}d\sigma_sdR \\
    &= V_tdt + \frac{1}{2}V_{ss}\sigma_sS^2dt+\frac{1}{2}V_{xx}\sigma_x^2dt + \frac{1}{2}V_{\sigma_s\sigma_s}\eta^2\sigma_sdt + \frac{1}{2}V_{rr}\sigma_r^2dt\\
    &\quad + V_{sx}\rho_x\sigma_x\sqrt{\sigma_s}Sdt + V_{s\sigma_s}\rho_s\eta\sigma_sSdt + V_{sr}\rho_r\sigma_r\sqrt{\sigma_s}Sdt 
    \\ & \quad
    + V_{x\sigma_s}\rho_x\rho_s\sigma_x\eta\sqrt{\sigma_s}dt + V_{xr}\rho_x\rho_r\sigma_x\sigma_rdt + V_{\sigma_sr}\rho_s\rho_r\eta\sigma_r\sqrt{\sigma_s}dt.
    \end{split}
\end{align*}

Equating $dt$ term and rearranging terms, we get 
\begin{align*} \begin{split} 
    \quad \ V_t &= R(V-SV_s-XV_x-\sigma_sV_{\sigma_s}-RV_r)
    - \frac{1}{2}\sigma_sS^2 V_{ss} - \frac{1}{2}\sigma_x^2 V_{xx} - \frac{1}{2}\eta^2\sigma_s V_{\sigma_s\sigma_s}- \frac{1}{2}\sigma_r^2V_{rr}- \rho_x\sigma_x\sqrt{\sigma_s}SV_{sx} \\
    & \quad - \rho_s\eta\sigma_sSV_{s\sigma_s} - \rho_r\sigma_r\sqrt{\sigma_s}SV_{sr} - \rho_x\rho_s\sigma_x\eta\sqrt{\sigma_s}V_{x\sigma_s} - \rho_x\rho_r\sigma_x\sigma_rV_{xr} - \rho_s\rho_r\eta\sigma_r\sqrt{\sigma_s}V_{\sigma_sr}, 
\end{split}
\end{align*}
which is the PDE \ref{MainPDE}. 
Next we show that the same PDE can be found using the risk-neutral pricing formula.

 \noindent {\textbf{Risk-neutral approach:}}
We define the discount process $$D(t) = e^{-\int R(t)dt} \quad \Longleftrightarrow \quad dD(t) = - R(t)D(t) dt.$$
%
We apply Ito's lemma to the discounted stock price process to find 
\begin{align*}
    \begin{split}
        d(DS) &= D dS + S dD + dDdS\\
        &= D ( (\mu + X +R)Sdt + \sqrt{\sigma_s}SdW_1) + S(- RDdt) + (- RDdt)( (\mu + X +RSdt + \sqrt{\sigma_s}SdW_1) \\
         &= D( (\mu + X +R)Sdt + \sqrt{\sigma_s}SdW_1) - S RDdt \\
        &= DS(\mu+X)dt +DS \sqrt{\sigma_s}dW_1\\
        &= DS\sqrt{\sigma_s}\left(\frac{(\mu+X)}{\sqrt{\sigma_s}}dt + dW_1\right)
    \end{split}
\end{align*}
Similarly, we make the choice that  
\begin{align*}
    d \widetilde{W}_2 &= \left( \frac{(-\kappa_x - R)X}{\sigma_x}- \frac{\rho_x (\mu + X)} {\sqrt{1-\rho_x^2}\sqrt{\sigma_s}}    \right)dt +dW_2, \\
   d\widetilde{W}_3 &= \left(\frac{\kappa_s(\sigma - \sigma_s) - R\sigma_s}   {\eta \sqrt{\sigma_s} \sqrt{1-\rho_s}} -  \frac{\rho_s (\mu + X)}{\sqrt{1-\rho_s^2} \sqrt{\sigma_s}} \right) dt + dW_3, \\
   d\widetilde{W}_4 &=   \left(\frac{(\kappa_r (r-R) - R^2)}{\sigma_r \sqrt{1-\rho_r^2}} - \frac{\rho_r (\mu + X)}{\sqrt{1-\rho_r^2} \sqrt{\sigma_s}} \right)dt + dW_3,
\end{align*}
Note: we did not need to make this choice, in fact, another common choice is that $\widetilde W_2 = W_2$, $\widetilde W_3 = W_3$ and $\widetilde W_4 = W_4$. 

We then apply Girsanov's theorem. Under the new measure $\widetilde{\mathbb{P}}$, $d\widetilde{W_1} =\frac{(\mu+X)}{\sqrt{\sigma_s}}dt + dW_1$ is the differential of a Brownian motion. Thus, 
    $$d(DS) = DS\sqrt{\sigma_s}d\widetilde{W_1}.$$
    This implies the discounted stock process $DS$ is an Ito integral, and thus, a martingale under $\widetilde{\mathbb{P}}.$ 
 
Under our choice, we also find that the discounted change in equity premium price process $d(DX)$, the discounted variance price process $d(D\sigma_s)$, and the discounted interest rate price $d(DR)$ are martingales under new measure $\widetilde{\mathbb{P}}.$
Then, we have
\begin{align} \label{newmeasure}
    \begin{cases}
        dS= S \sqrt{\sigma_s} d \widetilde{W}_1 +RS dt\\
        dX = \sigma_x(\rho_x d\widetilde{W}_1 + \sqrt{1- \rho_x^2} d\widetilde{W}_2)+ RX dt\\
        d\sigma_s =  \eta \sqrt{\sigma_s} (\rho_s d\widetilde{W}_1 + \sqrt{1- \rho_s^2}d\widetilde{W}_3) + R \sigma_s dt\\
        dR = \sigma_r (\rho_r d\widetilde{W}_1 + \sqrt{1-\rho_r^2}d\widetilde{W}_4)+ R^2 dt
    \end{cases}
\end{align}
Applying Ito's lemma to $DV$ and setting the $dt$ term equal to 0, we get 
\begin{align*} \begin{split} 
    \quad \ V_t &= R(V-SV_s-XV_x-\sigma_sV_{\sigma_s}-RV_r)
    - \frac{1}{2}\sigma_sS^2 V_{ss} - \frac{1}{2}\sigma_x^2 V_{xx} - \frac{1}{2}\eta^2\sigma_s V_{\sigma_s\sigma_s}- \frac{1}{2}\sigma_r^2V_{rr}- \rho_x\sigma_x\sqrt{\sigma_s}SV_{sx} \\
    & \quad - \rho_s\eta\sigma_sSV_{s\sigma_s} - \rho_r\sigma_r\sqrt{\sigma_s}SV_{sr} - \rho_x\rho_s\sigma_x\eta\sqrt{\sigma_s}V_{x\sigma_s} - \rho_x\rho_r\sigma_x\sigma_rV_{xr} - \rho_s\rho_r\eta\sigma_r\sqrt{\sigma_s}V_{\sigma_sr}, 
\end{split}
\end{align*}
which is the PDE \ref{MainPDE}.
\end{proof}

When we add equation \ref{Asianeq} to the system of SDEs \ref{mainmodel}, we get a PDE for the Asian option, as summarized in our second theorem. The proof is similar to the proof for Theorem \ref{ourPDE} and is omitted. 
\begin{thm}
Let $V$ be price of a European style derivative on a stock price process determined by system \ref{mainmodel}, and let $I$ be the process determined by equation \ref{Asianeq}. The function $V$ must be a solution to 
\begin{align} \begin{split}  \label{AsianPDE}
    \quad \ V_t &= R(V-SV_s-XV_x-\sigma_sV_{\sigma_s}-RV_r)
    -  V_I S - \frac{1}{2}\sigma_sS^2 V_{ss} - \frac{1}{2}\sigma_x^2 V_{xx} - \frac{1}{2}\eta^2\sigma_s V_{\sigma_s\sigma_s}- \frac{1}{2}\sigma_r^2V_{rr}\\
    &\quad - \rho_x\sigma_x\sqrt{\sigma_s}SV_{sx} - \rho_s\eta\sigma_sSV_{s\sigma_s} - \rho_r\sigma_r\sqrt{\sigma_s}SV_{sr} \\&\quad - \rho_x\rho_s\sigma_x\eta\sqrt{\sigma_s}V_{x\sigma_s} - \rho_x\rho_r\sigma_x\sigma_rV_{xr} - \rho_s\rho_r\eta\sigma_r\sqrt{\sigma_s}V_{\sigma_sr}
\end{split}
\end{align}
where $V_I$ is the first-order derivative of $V$ with respect to $I$.
\end{thm}
 In order to price a particular option, these PDEs must be pared with appropriate terminal values (data given at  time $T>0$). Under mild assumptions on the terminal values, it is known that the classical solutions to the corresponding Cauchy problems exist, are unique, and are smooth for all $t<T$. These results can be proved using the method of sub and super solutions, see for instance Lieberman \cite{lieberman_second_2005} for a detailed description of these techniques.

\section{Numerical results} \label{numerical}
In order to approximate the value of particular options, we choose to estimate the solutions of the above PDEs numerically using finite difference methods. Finite difference methods are efficient for parabolic equations, as the solution remains smooth. However, some numerical algorithms will require the size of the time step satisfy a Courant–Friedrichs–Lewy (CFL) condition \cite{courant_partial_2018}, which may require a large number of time steps. In order to circumvent this, we will also consider implicit algorithms which do not require the size of the time step satisify a CFL condition.  

We will sometimes use ``big-O'' notation, which we define now.
Let $f$, the function to be estimated, be a real or complex valued function and let $g$, the comparison function, be a real valued function. Let both functions be defined on some unbounded subset of the positive real numbers, and $g(x)$ be strictly positive for all large enough values of $x$. One writes
$$f(x) = O(g(x))\quad \text{as } x \rightarrow \infty$$
if the absolute value of $f(x)$ is at most a positive constant multiple of $g(x)$ for all sufficiently large values of $x$.  
\subsection{Finite difference methods}
\quad We construct a four-dimensional array that spans the interval $$[0, S_{\max}] \times [0, \sigma_{s\max}] \times [-X_{\max}, X_{\max}] \times [-R_{\max}, R_{\max}].$$ 

Let $\Delta t$, $\Delta S$, $\Delta \sigma_s$, and $\Delta X$ be the change in time, the change in stock price, the change in variance of the stock, the change in change in equity premium, and the change in interest rate. Let $V_{i, j, m, n}^h$ be the price of options at time $h$ with $i, j, m,$ and $n$ are the indexes stock price, variance, equity premium, and interest rate. The time index $h$ refers to time $T-h\Delta t$. We will go backwards in time to find the solution and at each time $T-h\Delta t$, the goal is to find the value of the option at time $T-h\Delta t - \Delta t$. 
\subsubsection{Forward Euler method} 
\quad In the forward Euler method, the derivative at the current time step is estimated using the information from the current time step itself. The update formula is: $V^{h+1} = V^h + \Delta t  f(V^h, t)$, where $\Delta t$ is the time step size, $V^h$ is the solution at time $t$, and $f(V^h, t)$ is the derivative of the function at time t.\\

\begin{itemize}
    \item First-order time derivative estimation
        \begin{align*}
            \frac{\partial V}{\partial t} \approx \frac{V^{h+1}_{i,j, m, n}-V^{h}_{i, j, m, n}}{\Delta t}
        \end{align*}
    \item First-order single-variable spatial derivative estimation: $$\frac{\partial U}{\partial S} \approx \frac{U^{t}_{i+1, j,m,n}-U^{t}_{i-1,j,m,n}}{2\Delta S},$$ and similar for $\cfrac{\partial U}{\partial \sigma_s}, \cfrac{\partial U}{\partial X},$ and $\cfrac{\partial U}{\partial R}.$ 
    \item Second-order single-variable spatial derivative estimation: $$\frac{\partial^2 U}{\partial S^2} \approx \frac{U^{t}_{i+1,j,m,n}-2U^{t}_{i, j,m,n}+U^{t}_{i-1, j,m,n}}{\Delta S^2},$$ and similar for $\cfrac{\partial^2 U}{\partial \sigma_s^2}, \cfrac{\partial^2 U}{\partial X^2},$ and $\cfrac{\partial^2 U}{\partial R^2}.$ 
    \item Second-order mixed-variable spatial derivative estimation
    $$\frac{\partial^2 V}{\partial S \partial \sigma_s} \approx \frac{V^{h}_{i+1,j+1,m,n}+V^{h}_{i-1,j-1,m,n}-V^{h}_{i+1,j-1,m,n}-V^{h}_{i-1,j+1,m,n}}{4\Delta S \Delta \sigma_s},$$
    and similar for $\cfrac{\partial^2 V}{\partial S \partial X},$ $\cfrac{\partial^2 V}{\partial S \partial R},$ $\cfrac{\partial^2 V}{\partial S \partial R},$ $\cfrac{\partial^2 V}{\partial \sigma_s \partial X},$ $\cfrac{\partial^2 V}{\partial \sigma_s \partial R},$ and $\cfrac{\partial^2 V}{\partial X \partial R}$.
            
            
            
            
\end{itemize}
\quad The forward Euler method formula is obtained by substituting the derivative estimations mentioned above into the PDE \ref{MainPDE}
$$V^{h+1}_{i, j,m,n} - V^{h}_{i, j,m,n} = \sum_{a, b, c, d\in \{-1,0,1\}} C_{a,b,c, d}V_{i+a,j+b,m+c, n+d}^t$$
for some coefficient $C_{a,b,c, d}$ from the PDE. Let $M$ be the matrix transformation such that $$M V^h = \sum_{a, b, c, d\in \{-1,0,1\}} C_{a,b,c, d}V_{i+a,j+b,m+c, n+d}^t.$$ 
Then, we get
\begin{equation} \label{eq1}
V^{h+1} - V^h = MV^h
\end{equation}
\subsubsection{Backward Euler method} 
\quad In the backward Euler method, the derivative at the next time step is estimated using the information from the next time step itself. The update formula is: $V^{h+1} = V^h + \Delta t  f(V^{h+1}, t)$. Similar to forward Euler method, we get:
\begin{equation} \label{eq2}
V^{h+1} - V^h = MV^{h+1}
\end{equation}
\subsubsection{Crank-Nicolson method} 
At each time step, the derivative terms in the PDE are approximated using a combination of values from the current time step and the next time step. 
The Crank-Nicolson method combines the forward Euler method in \ref{eq1} and backward Euler method in \ref{eq2} with a weighting parameter $\theta$ which is often set to be $0.5$:
\begin{equation} \label{cneq}
V^{h+1} - V^h = (1-\theta)MV^h + \theta MV^{h+1}
\end{equation}

Due to the difficulty of dealing with mixed derivative terms in an implicit method, we only use those terms explicitly. Let $M_{s\sigma_s}$, $M_{sx}$, $M_{sr}$, $M_{\sigma_sx}$, $M_{\sigma_sr}$, and $M_{x r}$ be the coefficient matrices of second-order mixed-variable spatial derivatives. Let $M_{ss}$, $M_{\sigma_s\sigma_s}$, $M_{\sigma_s r}$, and $M_{x x}$ be the coefficient matrices of both first-order and second-order single-variable spatial derivatives. The term $rV$ is split distributed evenly over $M_{ss}$, $M_{\sigma_s\sigma_s}$, $M_{\sigma_s r}$, and $M_{x x}$ as in Haentjens \& in’t Hout (2012) \cite{haentjens_adi_2011}. Similar to Lin \& Reisinger (2008) \cite{sensen_lin_christoph_reisinger_finite_2008}, let $$A = M_{s\sigma_s}+M_{sx}+M_{sr}+M_{\sigma_sx}+M_{\sigma_sr}+M_{x r},$$ and $$B = M_{ss}+M_{vv}+M_{xx}+M_{rr}.$$
Then, equation \ref{cneq} is equivalent to
        $$V^{h+1}-V^h = ((1-\theta)B + A)V^h + \theta B V^{h+1}  \quad  \iff \quad (I - \theta B)(V^{h+1}-V^h) = (A +B)V^h .$$
Pluging in $B = M_{ss}+M_{vv}+M_{xx}+M_{rr}$, we get 
$$(I - \theta (M_{ss}+M_{vv}+M_{xx}+M_{rr}))(V^{h+1}-V^h) = (A+B)V^h.$$ 
\vspace{0.5cm}
\noindent Since the terms $V^{h+1}-V^h$, $M_{ss}$, $M_{vv}, M_{xx}$, and $M_{rr}$ are $O(\Delta t)$, we have
\begin{align*}
(I-\theta M_{ss})(I-\theta M_{vv})(I-\theta M_{xx})(I-\theta M_{rr})(V^{h+1}-V^h)
= (A+B)V^h + O(\Delta t^3).
\end{align*}
Using the estimation  
$$(I-\theta M_{ss})(I-\theta M_{vv})(I-\theta M_{xx})(I-\theta M_{rr})(V^{h+1}-V^h) = (A+B)V^h + O(\Delta t^3),$$
we can compute $V^{h+1}$ using a splitting algorithm as follows
\begin{align*}
(I-\theta M_{ss})Y_1 &= (A + B)V^h \\
(I-\theta M_{vv})Y_2 &= Y_1 \\
(I-\theta M_{xx})Y_3 &= Y_2 \\
(I-\theta M_{rr})Y_4 &= Y_3 \\
V^{h+1} &= V^h + Y_4.
\end{align*}
When $\theta = 0$ and $\theta = 1$, this algorithm gives the solution for forward Euler method and backward Euler method respectively.
\subsection{European call option}
\subsubsection{Terminal and boundary conditions}
\quad Let $K$ be the strike price of the European call option. At the terminal time $T$, $V(s,\sigma_s, x, r, T) = \max\{s-K, 0\}.$ \\
The following boundary conditions are imposed for all $0 \leq t \leq T$:
\begin{subequations}
\begin{align}
    V &= 0 &&\quad  \text{when } s = 0, \label{s=0}\\
    V_s &= 1 &&\quad  \text{when } s \rightarrow \infty, \label{smax}\\
    rV &= V_t +  rs V_s + rx V_x + r^2V_r + \frac{1}{2} (V_{rr}{\sigma_r}^2 +V_{xx}{\sigma_x}^2) + V_{vr} \sigma_x \sigma_r \rho_r \rho_x &&\quad \text{when } \sigma_s = 0,\label{v=0} \\
    V &= s &&\quad  \text{when } \sigma_s\rightarrow \infty, \label{vmax}\\
    V &= 0 &&\quad  \text{when } x \rightarrow -\infty, \label{xmin}\\
    V_s &= 1 &&\quad  \text{when } x \rightarrow \infty, \label{xmax}\\
    V_r &= 0  &&\quad  \text{when } r = \pm R_{\max}. \label{rminmax}
\end{align}
\end{subequations}
\quad Condition \ref{s=0} is obvious, since when the underlying asset value is $0,$ the call option for this asset is worthless. Conditions \ref{smax} and \ref{vmax} are stated in Heston (1993) \cite{heston_closed-form_1993}. When $\sigma_s = 0$, we plug in $\sigma_s = 0$ into the PDE \ref{MainPDE} and obtain the condition \ref{v=0}. When $x \rightarrow -\infty$, the underlying asset price approaches 0, so the option price is 0 as in condition \ref{xmin}. When $x \rightarrow \infty$, the underlying asset price gets very large and also approaches infinity, so we get condition \ref{xmax} similar to condition \ref{smax}.
Condition \ref{rminmax} is used in Haentjens \& in’t Hout (2012) \cite{haentjens_adi_2011}.
\subsubsection{Numerical experiments} \label{call}
\begin{enumerate}[label=(\textbf{\alph*})]
    \item \textbf{Setting up parameters}\\
    Let $K = 5, \rho_s = 0.18, \rho_x = 0.23, \rho_r = 0.21, \eta = 0.027, \sigma_x = 0.011,$ and $\sigma_r = 0.019$. The graphs resulting from the given parameters and the Crank-Nicolson method are displayed below. Given our limitation to visualizing up to three dimensions, and considering that the option value at terminal time t=0 depends on four variables, we need to hold two variables constant while visualizing the option value based on the remaining two variables.
    \begin{figure}[!htb]
        \centering 
        \subfigure[]{\includegraphics[width=0.3\textwidth]{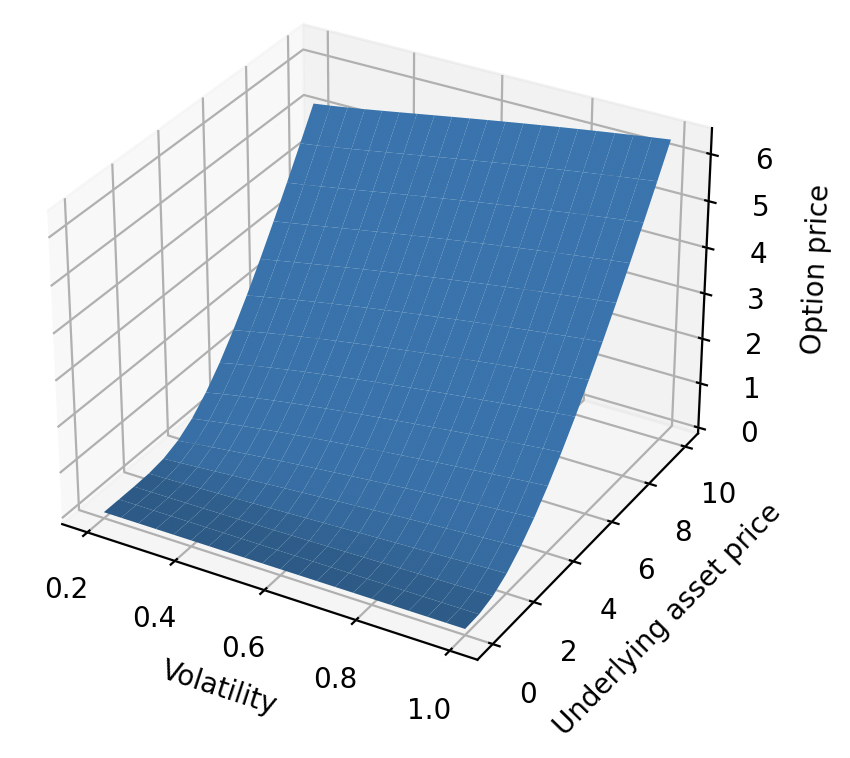}}
        \subfigure[]{\includegraphics[width=0.3\textwidth]{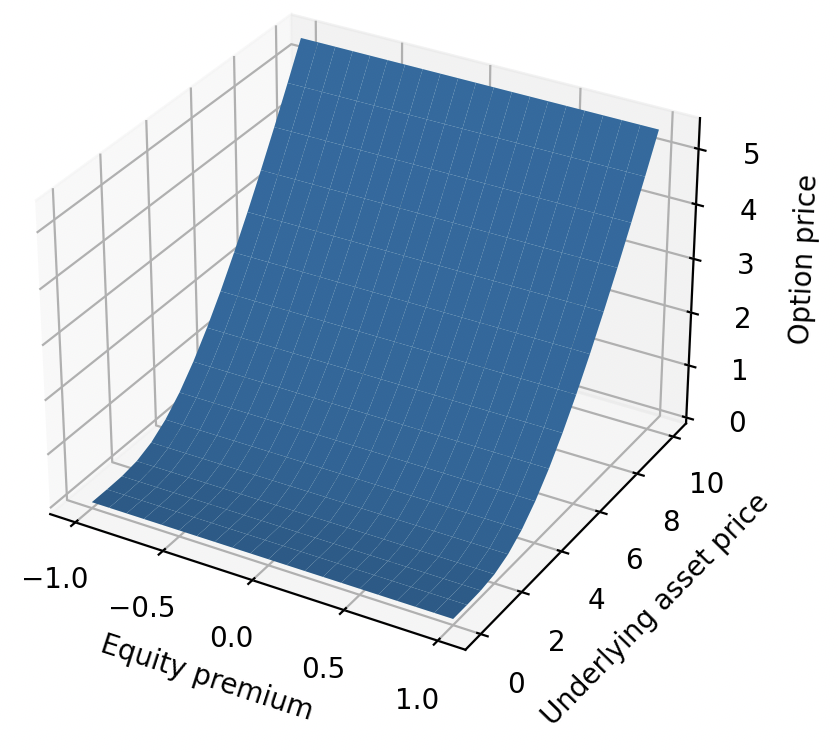}} 
        \subfigure[]{\includegraphics[width=0.3\textwidth]{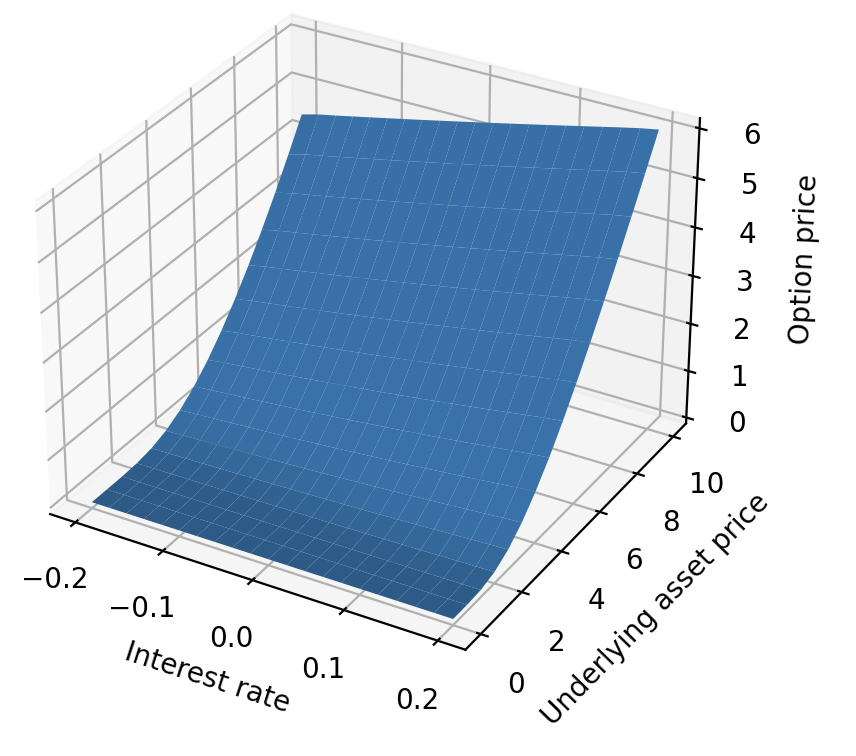}} 
        \caption{European option price plot of (a) $S\sigma_s$ slice when $X = 0.5$ and $R = 0.04$, (b) SX slice when $\sigma_s = 0.36$ and $R=0.4$, (c) SR slice when $\sigma_s = 0.36$ and $X = 0.5$}
    \end{figure}
    \FloatBarrier
    \item \textbf{Numerical comparison} \\
    We perform numerical experiments for three methods: forward Euler, backward Euler, and Crank-Nicolson. Through experimentation, we have determined that both backward Euler and Crank-Nicolson consistently yield stable results with 10 time steps. The forward Euler method needs 220 time steps for stability to be achieved. \\
    For comparison, each method will be executed using 220 time steps, and the outcomes for a given stock price, volatility, change in equity premium, and interest rate from each method is recorded in the table below. The table shows that there is a marginal discrepancy among the results obtained through the three methods. However, the differences remain relatively small. 
    \begin{table}[!htb] \label{comparetable}
        \centering
        \begin{tabular}{|c|c|c|c|c|c|c|}
        \cline{1-7}
        S  & V  & X & R & Forward Euler &  Backward Euler &Crank-Nicolson\\
        \hline
        $8$   & $0.28$   & $0.1$ & $0.02$ & $3.5171$ & $3.5127$ & $3.5149$  \\
        \hline
        $3.3$   & $0.4$   & $-0.3$ & $-0.16$ &$0.2743$ & $0.2738$ & $0.2741$  \\
        \hline
        $7.3$   & $0.8$   & $-0.6$ & $0.06$ &$3.2852$ & $3.2602$ & $3.2724$ \\
        \hline
        $6$   & $0.16$   & $-0.2$ & $0.1$ & $5.5657$ &$5.5266$ & $5.5458$ \\
        \hline
        \end{tabular}
        \caption{Results comparison among three methods}
        \label{time}
    \end{table}
     \FloatBarrier
    \item \textbf{Convergence} \\
    All three suggested approaches yield convergent solutions as the number of time steps increases. This implies that the finite difference approximation of the PDE approaches its actual solution. For each method, we increase the number of time steps between the initial time $t=0$ and the terminal time $t=T$, and then plot the option value against specified parameters including stock price, volatility, change in equity premium, and interest rate. The graph presented below illustrates the European call option prices at time $t=0$, computed using various numbers of time steps and different numerical methods, assuming $S = 8.5, \sigma_s = 0.28, X = 0,$ and $R = 0.02$. As the number of time steps increases, a convergence of option values is evident across all methods. In backward Euler, the difference in option price is much smaller than other two methods as the number of time steps increases and the values still signifies convergence.

    \begin{figure}[!htb] \label{convergence}
        \centering 
        \subfigure[]{\includegraphics[width=0.3\textwidth]{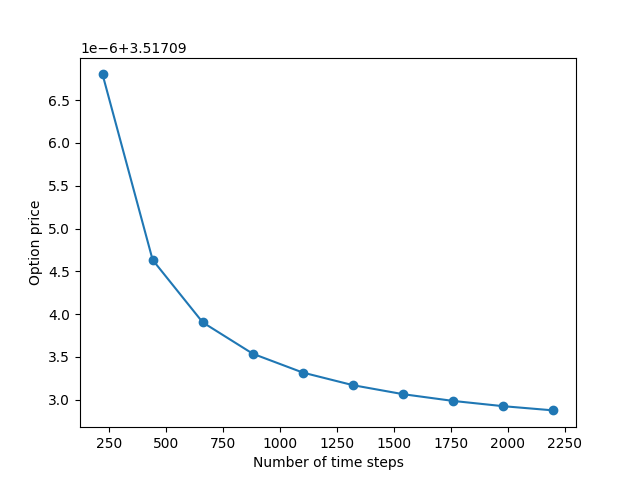}} 
        \subfigure[]{\includegraphics[width=0.3\textwidth]{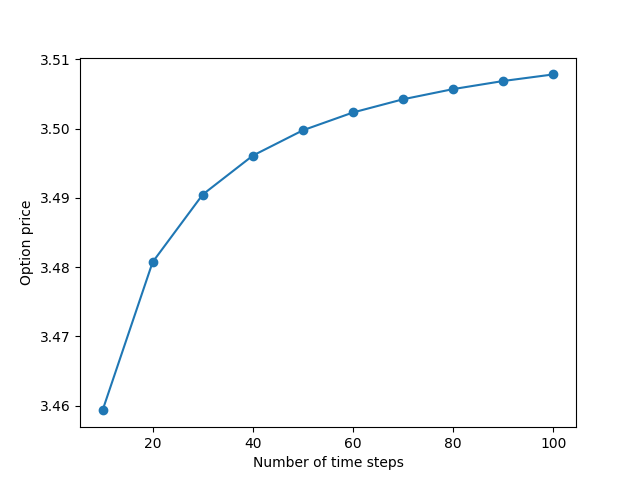}} 
        \subfigure[]{\includegraphics[width=0.3\textwidth]{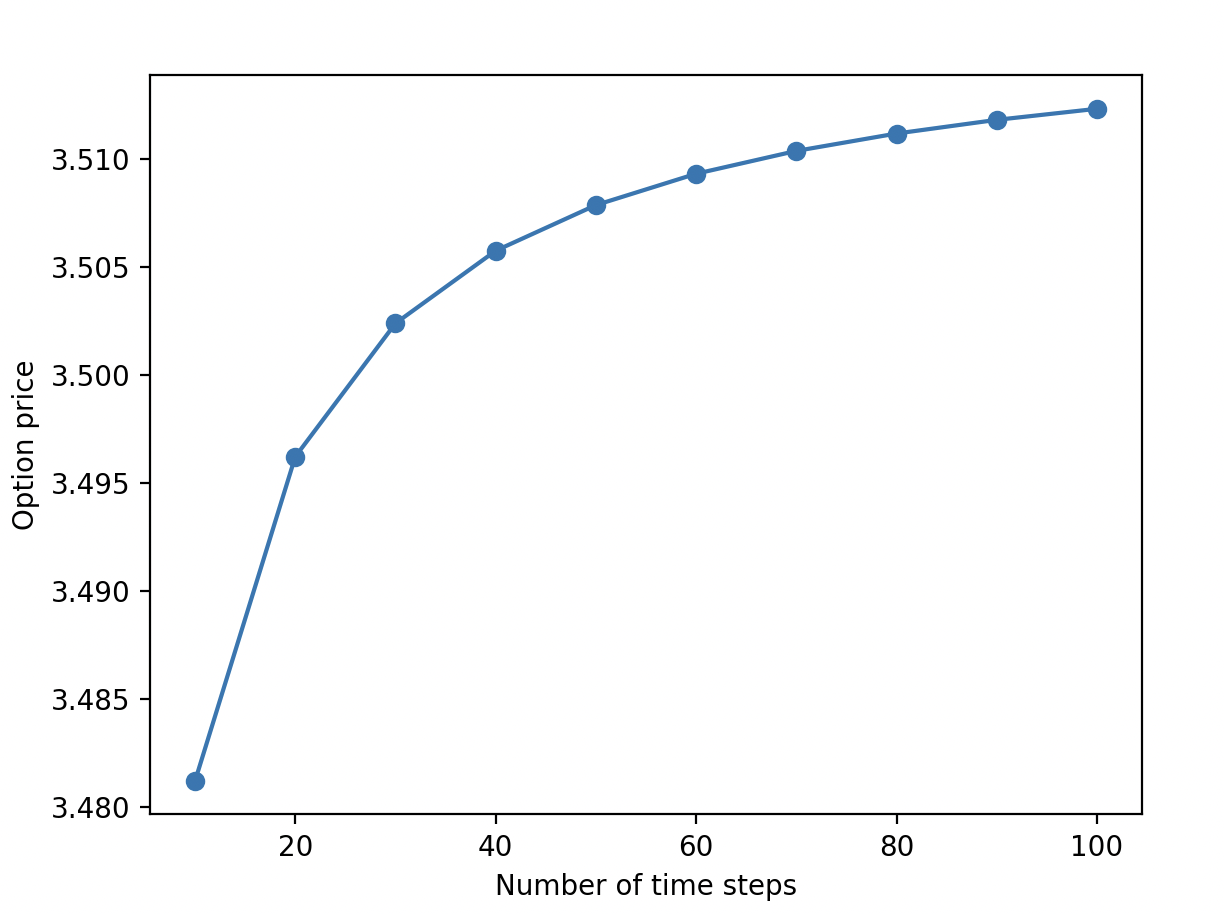}}
        \caption{Convergence of (a) Forward Euler method (b) Backward Euler method (c) Crank-Nicolson method}
    \end{figure}
    \FloatBarrier
\end{enumerate}
\subsection{European up-and-out call option}
\subsubsection{Terminal and boundary conditions}
\quad Let $B$ be the barrier price of the European up-and-out call option. The conditions \ref{s=0}, \ref{smax}, \ref{v=0}, \ref{xmin}, \ref{xmax} and \ref{rminmax} from European call option also hold for European up-and-out call option. The condition \ref{vmax} is replaced by 
$$V_{\sigma_s} = 0,$$
which is used in Haentjens \& in’t Hout (2012) \cite{haentjens_adi_2011}. By the definition of this type of option, 
$$V = 0 \quad \text{when } s \geq B, \forall 0 \leq t \leq T.$$ 
\subsubsection{Numerical experiments}
\quad Using the same parameters for European call option in \ref{call} and Crank-Nicolson method, we get graphs as displayed below for barrier $B = 8$. 
 \begin{figure}[!htb]
        \centering 
        \subfigure[]{\includegraphics[width=0.3\textwidth]{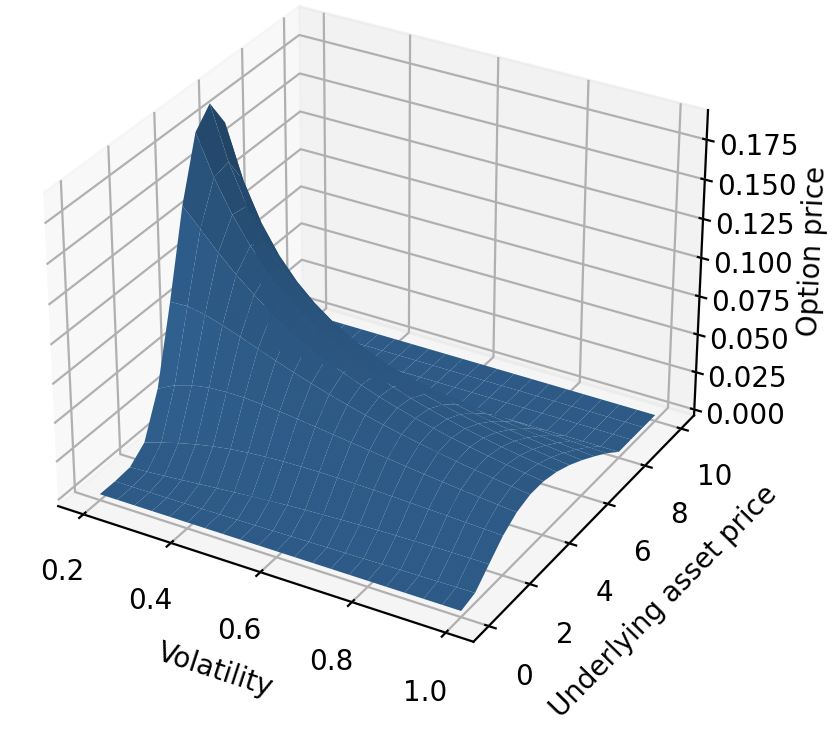}}
        \subfigure[]{\includegraphics[width=0.3\textwidth]{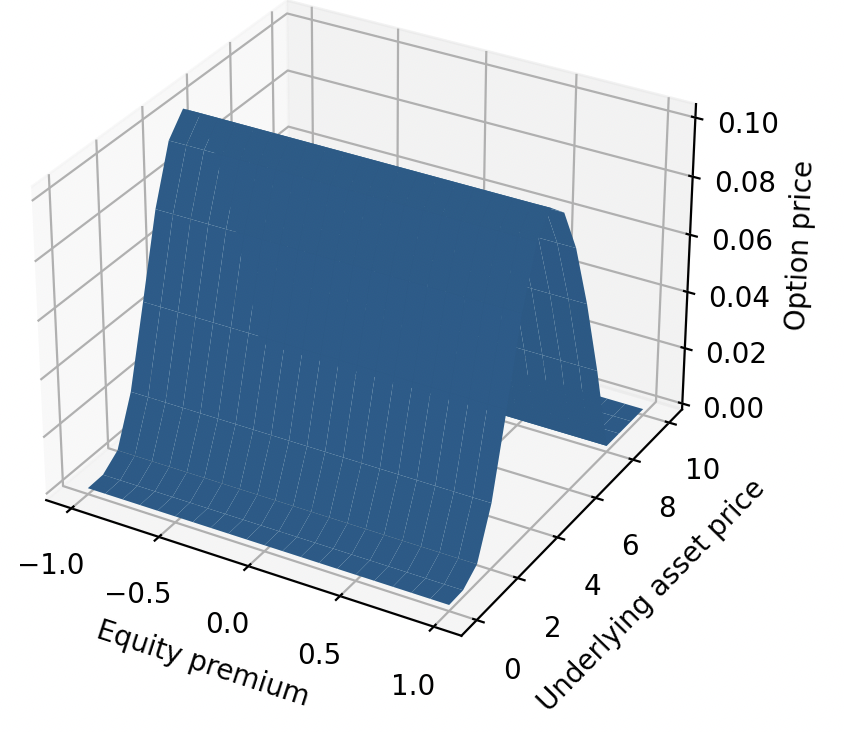}} 
        \subfigure[]{\includegraphics[width=0.3\textwidth]{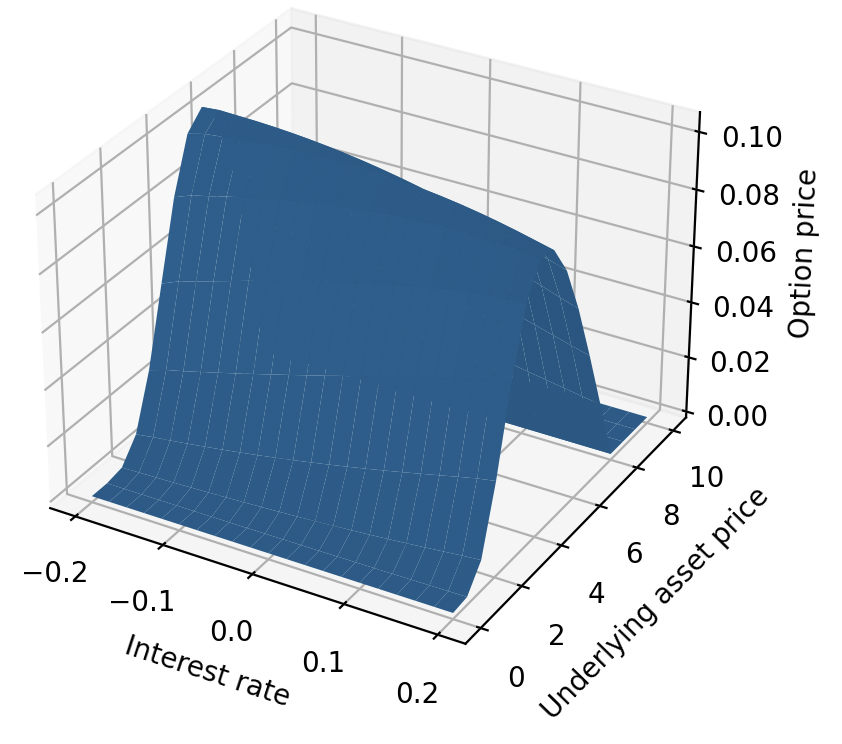}} 
        \caption{European up-and-out call option price plot of (a) $S\sigma_s$ slice when $X = 0.5$ and $R = 0.04$, (b) SX slice when $\sigma_s = 0.36$ and $R=0.4$, (c) SR slice when $\sigma_s = 0.36$ and $X = 0.5$}
    \end{figure}
    \FloatBarrier

\section{Acknowledgments}
This research project was done as part of the 2023 Ohio State ROMUS (Research Opportunities in Mathematics for Underrepresented Students) program. We thank Professor John Holmes for his guidance. We also thank the support from the NSF LEAPS DMS grant number 2247019. The third author Luong-Le would like to thank Christian Altamirano for helpful discussions on the numerical estimation for the Heston model.

\bibliographystyle{plain}

\bibliography{ROMUS.bib}
\end{document}